\documentclass[11pt,a4paper]{article}
\usepackage{amssymb,latexsym,amsmath}
\usepackage[dvipdfmx]{graphicx}
\usepackage{enumerate}
\usepackage{mysty2e_eng}
\usepackage{color}
\usepackage{txfonts}

\DeclareMathOperator*{\argmin}{arg\,min}

\definecolor{brown}{rgb}{0.8,0.4,0}
\definecolor{purple}{rgb}{0.49,0.18,0.56}

\begin{document}

\def\varliminf{\mathop{\underline{\rule[-0.02em]{0em}{0.2em}
 \hbox{\rm lim}}}}
\def\varlimsup{\mathop{\overline{\hbox{\rm lim}}}}
\def\esssup{\mathop{\hbox{\rm ess\,sup}}}


\newcounter{thebrenum}
\newenvironment{brenum}%
{\begin{list}{(\arabic{thebrenum}) }%
{\usecounter{thebrenum}
\setcounter{thebrenum}{0}
\setlength{\labelsep}{0ex}
\setlength{\labelwidth}{4.1ex}
\setlength{\leftmargin}{\labelwidth}
\setlength{\itemsep}{0ex}
\setlength{\parsep}{0.5ex}
}}%
{\end{list}}


\newcounter{thebrenumalph}
\newenvironment{brenumalph}%
{\begin{list}{(\alph{thebrenumalph}) }%
{\usecounter{thebrenumalph}
\setcounter{thebrenumalph}{0}
\setlength{\labelsep}{0ex}
\setlength{\labelwidth}{4.1ex}
\setlength{\leftmargin}{\labelwidth}
\setlength{\itemsep}{0ex}
\setlength{\parsep}{0.5ex}
}}%
{\end{list}}


\newcounter{thebrenumroman}
\newenvironment{brenumroman}%
{\begin{list}{(\roman{thebrenumroman}) }%
{\usecounter{thebrenumroman}
\setcounter{thebrenumroman}{0}
\setlength{\labelsep}{0ex}
\setlength{\labelwidth}{3.3ex}
\setlength{\leftmargin}{\labelwidth}
\setlength{\itemsep}{0ex}
\setlength{\parsep}{0.7ex}
}}%
{\end{list}}


\newcommand{\bda}{\boldsymbol{a}}
\newcommand{\bdb}{\boldsymbol{b}}
\newcommand{\bdc}{\boldsymbol{c}}
\newcommand{\bdd}{\boldsymbol{d}}
\newcommand{\bde}{\boldsymbol{e}}
\newcommand{\bdf}{\boldsymbol{f}}
\newcommand{\bdg}{\boldsymbol{g}}
\newcommand{\bdh}{\boldsymbol{h}}
\newcommand{\bdi}{\boldsymbol{i}}
\newcommand{\bdj}{\boldsymbol{j}}
\newcommand{\bdk}{\boldsymbol{k}}
\newcommand{\bdl}{\boldsymbol{l}}
\newcommand{\bdm}{\boldsymbol{m}}
\newcommand{\bdn}{\boldsymbol{n}}
\newcommand{\bdo}{\boldsymbol{o}}
\newcommand{\bdp}{\boldsymbol{p}}
\newcommand{\bdq}{\boldsymbol{q}}
\newcommand{\bdr}{\boldsymbol{r}}
\newcommand{\bds}{\boldsymbol{s}}
\newcommand{\bdt}{\boldsymbol{t}}
\newcommand{\bdu}{\boldsymbol{u}}
\newcommand{\bdv}{\boldsymbol{v}}
\newcommand{\bdw}{\boldsymbol{w}}
\newcommand{\bdx}{\boldsymbol{x}}
\newcommand{\bdy}{\boldsymbol{y}}
\newcommand{\bdz}{\boldsymbol{z}}
\newcommand{\bdH}{\boldsymbol{H}}
\newcommand{\bdN}{\boldsymbol{N}}
\newcommand{\bdR}{\boldsymbol{R}}
\newcommand{\bdS}{\boldsymbol{S}}
\newcommand{\bdL}{\boldsymbol{L}}

\newcommand{\bdzero}{\boldsymbol{0}}
\newcommand{\bdone}{\boldsymbol{1}}

\newcommand{\bdalpha}{\boldsymbol{\alpha}}
\newcommand{\bdbeta}{\boldsymbol{\beta}}
\newcommand{\bdtheta}{\boldsymbol{\theta}}
\newcommand{\bdlambda}{\boldsymbol{\lambda}}
\newcommand{\bdxi}{\boldsymbol{\xi}}
\newcommand{\bdeta}{\boldsymbol{\eta}}
\newcommand{\bdzeta}{\boldsymbol{\zeta}}
\newcommand{\bdomega}{\boldsymbol{\omega}}

\newcommand{\rma}{{\mathrm a}}
\newcommand{\rmb}{{\mathrm b}}
\newcommand{\rmc}{{\mathrm c}}
\newcommand{\rmd}{{\mathrm d}}
\newcommand{\rme}{{\mathrm e}}
\newcommand{\rmf}{{\mathrm f}}
\newcommand{\rmg}{{\mathrm g}}
\newcommand{\rmh}{{\mathrm h}}
\newcommand{\rmi}{{\mathrm i}}
\newcommand{\rmj}{{\mathrm j}}
\newcommand{\rmk}{{\mathrm k}}
\newcommand{\rml}{{\mathrm l}}
\newcommand{\rmm}{{\mathrm m}}
\newcommand{\rmn}{{\mathrm n}}
\newcommand{\rmo}{{\mathrm o}}
\newcommand{\rmp}{{\mathrm p}}
\newcommand{\rmq}{{\mathrm q}}
\newcommand{\rmr}{{\mathrm r}}
\newcommand{\rms}{{\mathrm s}}
\newcommand{\rmt}{{\mathrm t}}
\newcommand{\rmu}{{\mathrm u}}
\newcommand{\rmv}{{\mathrm v}}
\newcommand{\rmw}{{\mathrm w}}
\newcommand{\rmx}{{\mathrm x}}
\newcommand{\rmy}{{\mathrm y}}
\newcommand{\rmz}{{\mathrm z}}

\newcommand{\rmA}{{\mathrm A}}
\newcommand{\rmE}{{\mathrm E}}
\newcommand{\rmF}{{\mathrm F}}
\newcommand{\rmG}{{\mathrm G}}
\newcommand{\rmH}{{\mathrm H}}
\newcommand{\rmL}{{\mathrm L}}
\newcommand{\rmM}{{\mathrm M}}
\newcommand{\rmN}{{\mathrm N}}
\newcommand{\rmO}{{\mathrm O}}
\newcommand{\rmR}{{\mathrm R}}
\newcommand{\rmT}{{\mathrm T}}
\newcommand{\rmU}{{\mathrm U}}

\newcommand{\vGamma}{\varGamma}
\newcommand{\vDelta}{\varDelta}
\newcommand{\vTheta}{\varTheta}
\newcommand{\vLambda}{\varLambda}
\newcommand{\vXi}{\varXi}
\newcommand{\vPi}{\varPi}
\newcommand{\vSigma}{\varSigma}
\newcommand{\vUpsilon}{\varUpsilon}
\newcommand{\vPhi}{\varPhi}
\newcommand{\vPsi}{\varPsi}
\newcommand{\vOmega}{\varOmega}

\newcommand{\bdvDelta}{{\boldsymbol{\varDelta}}}

\allowdisplaybreaks
\newcommand{\setunit}[1]{\setlength{\unitlength}{#1}}
\renewcommand{\hat}[1]{\widehat{#1}}
\renewcommand{\tilde}[1]{\widetilde{#1}}

\newcommand{\qed}{\hfill$\square$}
\newcommand{\fin}{\hfill$\Diamond$}
\newcommand{\etal}{{\em et al.}}
\newcommand{\ie}{{\em i.e.}}
\newcommand{\eg}{{\em e.g.}}

\newcommand{\Zset}{\mathbb Z}
\newcommand{\Rset}{\mathbb R}
\newcommand{\Cset}{\mathbb C}
\newcommand{\Msng}{\overline{\sigma}}
\newcommand{\msng}{\underline{\sigma}}
\newcommand{\Meig}{\overline{\lambda}}
\newcommand{\meig}{\underline{\lambda}}
\newcommand{\ol}[1]{\overline{#1}}
\newcommand{\ul}[1]{\underline{#1}}
\newcommand{\uind}[1]{u^{(#1)}}
\newcommand{\xind}[1]{x^{(#1)}}
\newcommand{\Nind}[1]{N^{(#1)}}

\newtheorem{thm}{Theorem}
\newtheorem{prop}{Proposition}
\newtheorem{lem}{Lemma}
\newtheorem{cor}{Corollary}
\newtheorem{problem}{Problem}
\newenvironment{prob}{\begin{problem}\upshape}{\end{problem}}
\newtheorem{algorithm}{Algorithm}
\newenvironment{algo}{\begin{algorithm}\upshape}{\end{algorithm}}
\newtheorem{example}{Example}
\newenvironment{exmpl}{\begin{example}\upshape}{\end{example}}
\newtheorem{remark}{Remark}
\newenvironment{rem}{\begin{remark}\upshape}{\end{remark}}
\newtheorem{assumption}{Assumption}
\newenvironment{assum}{\begin{assumption}\upshape}{\end{assumption}}

\title{\vspace*{-1.2cm}\bf\Large
Infinite-Horizon Sparse Optimal Control:
Solution through a Finite-Horizon Subproblem and
Its Receding-Horizon Implementation$^*$
\footnotetext{$\mbox{}^*$%
Released on \today.
This work is supported by the JSPS Kakenhi 23K03916 and
the Nanzan University Pache Research Subsidy
I-A-2 for the academic year 2026.
\vspace*{0.1cm}}%
\footnotetext{$\mbox{}^\dagger$%
Department of Mechanical Engineering and System Control,
Nanzan University,
Yamazatocho 18, Showa-ku, Nagoya 466-8673, Japan;
email: oishi@nanzan-u.ac.jp}%
\footnotetext{$\mbox{}^\ddagger$%
School of Informatics and Data Science, Hiroshima University,
Kagamiyama 1-4-1, Higashi-Hiroshima 739-8527, Japan}%
\footnotetext{$\mbox{}^\S$%
Graduate School of Advanced Science and Engineering,
Hiroshima University, 
Kagamiyama 1-4-1, Higashi-Hiroshima 739-8527, Japan}
\vspace*{-0.2cm}}
\author{\large Yasuaki Oishi$^\dagger$,
Takumi Iwata$^\ddagger$, and
Masaaki Nagahara$^\S$}
\date{\normalsize\vspace*{-0.9cm}}

\abstract{%
\small\setlength{\baselineskip}{14.5pt}\noindent
Sparse optimal control is considered in the infinite horizon.
In the literature, sparse control has been considered mostly in a finite horizon
for its formulation into a finite-dimensional optimization problem.
It is shown in this paper that an optimal solution of
the infinite-horizon sparse control problem
can be obtained through a solution of some finite-horizon subproblem.
This is due to sparsity of the optimal solution in the sense that
the optimal control input is constantly equal to zero at its tail.
An estimate is given on the horizon length required by this subproblem
and its adaptive choice is also discussed.
Implementation with a receding-horizon technique is considered
and its optimality and sparsity are guaranteed.
\\[0.2cm]%
\vspace*{-0.1cm}}

\maketitle
\thispagestyle{empty}

\section{Introduction}
\label{sec:intro}
Sparse optimal control is a control method suitable for saving energy.
Although conventional optimal control based on quadratic objective
function is mathematically favorable, it usually gives a dense
control input, which is nonzero for most of the time.
Using the $1$-norm of the input in the objective function,
we can enhance sparsity of an optimal input in the sense that
it is equal to zero for a long time duration.
Early results on sparse optimal control are found in \cite{AtF66,GaM12}.
In \cite{NQN16}, a sufficient condition was provided
for minimization of the sum of the input 1-norms to give maximum sparsity.
This result was further extended to a discrete-time system
\cite{NOQ16} and to an infinite-dimensional system \cite{IkN25}.
In \cite{Rao18}, the $1$-norm of the state was also introduced to
the objective function and a tradeoff was investigated
between the input sparsity and the state sparsity.
An efficient computation for solving sparse optimal control problems
was explored in \cite{AHW12,PoT20,DBQN22}.

In the literature, sparse optimal control has mostly been considered
in a finite horizon.
This is for reducing the problem to a finite-dimensional optimization
problem.
From a control perspective, however, it is not straightforward
to choose a finite horizon beforehand so that the control objective
is achieved there.
One way to apply a finite-horizon control input in the infinite horizon
is the use of the receding-horizon technique, that is,
repeatedly applying a finite-horizon control input with shifting its horizon.
This approach can be problematic however because
it may give a non-sparse control input as will be seen
in Example~\ref{exmpl:simple}.
In \cite{NQN16}, a self-triggered approach is proposed instead.
Although this approach gives a sparse control input,
its optimality in the infinite horizon is not clear.

In this paper,
infinite-horizon sparse optimal control is considered
in a discrete-time framework.
Here is a summary of its contribution.
\begin{enumerate}[(1)]
\item
It is shown that an optimal solution in the infinite horizon
can be obtained by solving some finite-horizon subproblem.
Indeed, if the horizon length is long enough,
an optimal control input of this subproblem is sparse
in the sense that it becomes zero after some finite time.
If one puts zeros at its tail,
the extended input remains optimal in a longer horizon.
The situation is the same even for its extension
to the infinite horizon.

\item
The horizon length required for the first subproblem is
estimated explicitly.
For the case where this estimate is conservative,
adaptive choice of the horizon length is also considered.

\item
The proposed approach can be combined with the receding-horizon technique.
The resulting control input is guaranteed to be
optimal and sparse in the infinite horizon.
\end{enumerate}

The rest of this paper is constructed as follows.
Section~\ref{sec:problem} gives an infinite-horizon optimal control
problem to be considered.
Section~\ref{sec:finite} discusses properties of
some finite-horizon optimal control problem.
Its optimal control input becomes optimal also in the
infinite-horizon if it is extended by zeros,
which is shown in Section~\ref{sec:solution}.
This property holds if the original horizon length is longer
than some specific number.
An adaptive way to choose a long enough horizon is discussed
in Section~\ref{sec:adaptive}.
A receding-horizon implementation is considered in Section~\ref{sec:receding}.
There, optimality and sparsity of the resulting control input are
guaranteed.
Section~\ref{sec:examples} provides two examples and
Section~\ref{sec:concl} concludes the paper.

Related results have been presented at conferences
\cite{OIN22,OIN23,OIN26} without proofs.
The present paper gives full proofs and first discusses
adaptive choice of the horizon length.
The spacecraft example in Section~\ref{sec:examples} also
first appears in this paper.
Moreover, a class of controlled plants is more general here
than in \cite{OIN23,OIN26} by the use of the factor $\alpha$.
Although this factor was used in \cite{OIN22}, too,
only a preliminary result was given there and
infinite-horizon optimality was not discussed.

The following notation is used.
For a vector $u=(u_1\ \ u_2\ \ \cdots\ \ u_m)^\rmT$,
its $1$-norm $\sum_{i=1}^m|u_i|$ and $\infty$-norm
$\max_{i=1, 2, \ldots, m}|u_i|$ are denoted by
$\|u\|_1$ and $\|u\|_\infty$, respectively.
For a matrix $B$, its norm induced by the $1$-norm, namely,
$\sup_{u\neq 0}\|Bu\|_1/\|u\|_1$, is expressed by $\|B\|_1$.
It is known that $\|B\|_1=\max_{i=1, 2, \ldots, m}\|b_i\|_1$
if $B$ consists of $m$ columns $b_1, b_2, \ldots, b_m$.
There hold $\|Bu\|_1\leq\|B\|_1\|u\|_1$ and
$\|p^\rmT B\|_\infty\leq\|p\|_\infty\|B\|_1$
for any vectors $u$ and $p$ having consistent dimensions.

\section{Considered Problem}
\label{sec:problem}
A plant to be controlled is a linear discrete-time system
\begin{equation}
x_\rmo(k+1)=A_\rmo x_\rmo(k)+B_\rmo u_\rmo(k)\ \ (k=0,\ 1,\ \ldots),\quad
x_\rmo(0)=\xi
\label{eq:original}
\end{equation}
with a state $x_\rmo(k)\in\Rset^n$ and an input $u_\rmo(k)\in\Rset^m$.
It is assumed that $(A_\rmo,B_\rmo)$ is stabilizable.

Let $\alpha\geq 1$ be some number such that $\alpha A_\rmo$ has
no eigenvalues on the unit circle in the complex plane.
Here the following infinite-horizon control problem is considered:
\begin{alignat*}{2}
P_\rmo:\ \ &\text{minimize}\ \ &&\sum_{k=0}^\infty\alpha^k\|u_\rmo(k)\|_1\\
&\text{subject to}\ \ &&x_\rmo(k+1)=A_\rmo x_\rmo(k)+B_\rmo u_\rmo(k)\ \ 
(k=0,\ 1,\ \ldots),\\
&&&x_\rmo(0)=\xi,\quad \alpha^k x_\rmo(k)\rightarrow 0\ \ (k\rightarrow\infty),\\
&&&\|u_\rmo(k)\|_\infty\leq 1\ \ (k=0,\ 1,\ \ldots).
\end{alignat*}

When $A_\rmo$ itself has no eigenvalues on the unit circle,
the standard choice $\alpha=1$ can be taken.
In this case,
the problem $P_\rmo$ is basically the same as the one considered
in the literature \cite{AtF66,NQN16,NOQ16,DBQN22}.
Indeed, the objective function is the sum of the input $1$-norms
in order to induce a sparse property of the input, \ie,
$u_\rmo(k)=0$ for many $k$'s.
The difference is that, while the problem has been considered in a finite
horizon in the literature, the problem $P_\rmo$ is in the infinite horizon.
Unless the time duration for control is fixed beforehand,
the infinite-horizon problem $P_\rmo$ appears more acceptable in practice.

When $A_\rmo$ has an eigenvalue on the unit circle,
the factor $\alpha$ has to be chosen larger than unity.
In this case, the higher penalty is posed on the input in the farther future
and exponential stability is asked on the closed-loop system.
This setting looks necessary in order to have the results
in this paper.
This issue will be discussed at the end of Section~\ref{sec:receding}.

With the transformation $x(k)=\alpha^k x_\rmo(k)$ and
$u(k)=\alpha^k u_\rmo(k)$,
the original plant dynamics \eqref{eq:original} is rewritten as
\begin{equation}
x(k+1)=Ax(k)+Bu(k)\ \ (k=0,\ 1,\ \ldots),\quad
x(0)=\xi
\label{eq:plant}
\end{equation}
for $A=\alpha A_\rmo$ and $B=\alpha B_\rmo$.
Obviously, $(A,B)$ is stabilizable and $A$ has no eigenvalues on
the unit circle.
The problem $P_\rmo$ is equivalently modified into
\begin{alignat*}{2}
P:\ \ &\text{minimize}\ \ &&\sum_{k=0}^\infty\|u(k)\|_1\\
&\text{subject to}\ \ &&x(k+1)=Ax(k)+Bu(k)\ \ 
(k=0,\ 1,\ \ldots),\\
&&&x(0)=\xi,\quad x(k)\rightarrow 0\ \ (k\rightarrow\infty),\\
&&&\|u(k)\|_\infty\leq\alpha^k\ \ (k=0,\ 1,\ \ldots).
\end{alignat*}
We consider the problem in this simpler form throughout this paper.

Even if the infinite-horizon problem is more natural from a control perspective,
it is not clear how to solve it unlike its finite-horizon counterpart.
A claim of this paper is that
an optimal solution of the problem $P$ can be obtained by
solving some finite-horizon subproblem,
which is introduced in the next section.

\begin{rem}
\label{rem:antistable}
If the matrix $A$ has its eigenvalues only inside the unit circle,
the problem $P$ (and the problem $P_\rmo$ as well)
is trivial and the zero input $u(k)=0$ $(k=0, 1, \ldots)$ is optimal.
Henceforth, we consider the case where $A$ has at least one eigenvalue
outside the unit circle.
\fin
\end{rem}

\begin{rem}
\label{rem:feasiblity}
The problem $P$ may not be feasible for an initial state $\xi$
distant from the origin due to boundedness of the input.
The region of $\xi$ for which the problem $P$ becomes feasible
is called the null controllable region and has been investigated
in \cite{HLQ02,HMQ02}.
We consider in the following the case where
$\xi$ is in this region and $P$ has a feasible solution
making the objective function finite.
\fin
\end{rem}

\begin{rem}
Although a solution of $P$ consists of an input $u(k)$ and a state
$x(k)$ for $k=0, 1, \ldots$,
our attention is paid on an input because
a state is uniquely determined for a given input.
We say a feasible input and an optimal input to mean
the input of a feasible solution and that of an optimal solution,
respectively.
\fin
\end{rem}

\section{Finite-Horizon Subproblem}
\label{sec:finite}
We introduce in this section some finite-horizon optimal control
problem of special importance in relation with the problem $P$.
Indeed, its optimal input turns out optimal also in $P$
if it is extended by putting zeros at its tail.

In order to introduce the finite-horizon problem, we first
decompose our plant \eqref{eq:plant}
into its stable part and anti-stable part.
That is, we consider a state transform
$x(k)=T(x_\rms(k)^\rmT\ \ x_\rma(k)^\rmT)^\rmT$
for some appropriate invertible matrix $T$ and
modify the representation \eqref{eq:plant} into the form
\[
\begin{pmatrix} x_\rms(k+1) \\ x_\rma(k+1) \end{pmatrix}=
\begin{pmatrix} A_\rms & O \\ O & A_\rma \end{pmatrix}
\begin{pmatrix} x_\rms(k) \\ x_\rma(k) \end{pmatrix}+
\begin{pmatrix} B_\rms \\ B_\rma \end{pmatrix}u(k)\ \
(k=0,\ 1, \ldots),\quad
\begin{pmatrix} x_\rms(0) \\ x_\rma(0) \end{pmatrix}=T^{-1}\xi,
\]
where $A_\rms$ has all of its eigenvalues inside the unit
circle and $A_\rma$ outside the unit circle.
The subsystem $(A_\rms,B_\rms)$ is called the {\em stable part} of the plant
and $(A_\rma,B_\rma)$ the {\em anti-stable part}.
The anti-stable part exists under the assumption made in Remark~\ref{rem:antistable}
while the stable part may not.
It is possible to assume $\|A_\rma^{-1}\|_1<1$
by an appropriate choice of $T$ (See Lemma~5.6.10 of \cite{HoJ85}).
The dimension of $x_\rma(k)$, the anti-stable part of the state, is
denoted by $n_\rma$.

We introduce the following finite-horizon problem for a positive integer $N$:
\begin{alignat*}{2}
F(N):\ \ &\text{minimize}\ \ &&\sum_{k=0}^{N-1}\|u(k)\|_1\\
&\text{subject to}\ \ 
&&x_\rma(k+1)=A_\rma x_\rma(k)+B_\rma u(k)\ \ (k=0,\ 1,\ \ldots,\ N-1),\\
&&&x_\rma(0)=(O\ \ I)T^{-1}\xi,\quad x_\rma(N)=0,\\
&&&\|u(k)\|_\infty\leq\alpha^k\ \ (k=0,\ 1,\ \ldots,\ N-1).
\end{alignat*}
Note that only the anti-stable part of the plant is considered
in the problem $F(N)$.
The positive integer $N$ is referred to as the {\em horizon length}.

On this finite-horizon subproblem $F(N)$, the following optimality condition
is available.
Here we write the $i$th component of $u(k)$ as $u_i(k)$ and
the $i$th column of $B_\rma$, the anti-stable part of the $B$ matrix, as $b_i$
for $i=1, 2, \ldots, m$.

\begin{lem}
\label{lem:optimality}
Suppose the subproblem $F(N)$ is feasible.
Then, the input $u^*(k)$ $(k=0,\ 1,\ \ldots,\ N-1)$ is optimal
in $F(N)$ if and only if
there exist a state $x_\rma^*(k)$ $(k=0,\ 1,\ \ldots,\ N)$ and
a costate $p_\rma^*(k)$ $(k=1,\ \ldots,\ N)$ satisfying
\begin{align}
x_\rma^*(k+1)&=A_\rma x_\rma^*(k)+B_\rma u^*(k)\ \
(k=0,\ 1,\ \ldots,\ N-1),\label{eq:opt_x}\\
x_\rma^*(0)&=(O\ \ I)T^{-1}\xi,\quad
x_\rma^*(N)=0,\label{eq:opt_x2}\\
p_\rma^*(k)^\rmT&=p_\rma^*(k+1)^\rmT A_\rma\ \
(k=1,\ 2,\ \ldots,\ N-1),\label{eq:opt_p}\\
u_i^*(k)&=\argmin_{|u_i(k)|\leq\alpha^k}\Bigl[
|u_i(k)|+p_\rma^*(k+1)^\rmT b_i u_i(k)\Bigr]\ \
(i=1,\ 2,\ \ldots,\ m;\ k=0,\ 1,\ \ldots,\ N-1).\label{eq:opt_u}
\end{align}
The input $u_i^*(k)$ satisfying the last equation has the following
properties:
\begin{alignat*}{2}
&p_\rma^*(k+1)^\rmT b_i<-1&&\text{\ \ implies\ \ }u_i^*(k)=\alpha^k;\\
&p_\rma^*(k+1)^\rmT b_i=-1&&\text{\ \ implies\ \ }0\leq u_i^*(k)\leq\alpha^k;\\
&-1<p_\rma^*(k+1)^\rmT b_i<1&&\text{\ \ implies\ \ }u_i^*(k)=0;\\
&p_\rma^*(k+1)^\rmT b_i=1&&\text{\ \ implies\ \ }-\alpha^k\leq u_i^*(k)\leq 0;\\
&p_\rma^*(k+1)^\rmT b_i>1&&\text{\ \ implies\ \ }u_i^*(k)=-\alpha^k.
\end{alignat*}
\end{lem}

\noindent
{\em Proof.}
Necessity is the well-known condition for optimality.
See Section~IV-D of \cite{NQN16} or Proposition~3.3.2 of \cite{Ber95a}.

Sufficiency follows from Proposition~3.3.4 of \cite{Ber95b}.
For completeness,
the proof is presented in a form adapted to the present context.
See Appendix~\ref{pf:optimality}.

The properties of $u^*_i(k)$ in the last statement easily follow
from the condition~\eqref{eq:opt_u}
\qed
\bigskip

Let $u(k)$ $(k=0, 1, \ldots, N-1)$ be a feasible input of $F(N)$.
Then, for any $N'>N$, we can have a feasible input of $F(N')$
by adding the zero input at the tail of $u(k)$, that is,
by defining $u(k)=0$ for $k=N, N+1, \ldots, N'-1$.
Indeed, since $x_\rma(N)=0$ and $u(k)=0$ after $k=N$,
the state $x_\rma(k)$ remains zero up to $k=N'$.
We call this extended input the {\em zero extension} of $u(k)$.

The next lemma claims that
an optimal input of the finite-horizon subproblem $F(N)$ is sparse
if the horizon length $N$ is long enough and
its zero extension is also optimal in a problem with a longer horizon.
To state the lemma, we need to define a positive integer $K$,
which is called the {\em zero-input time}.

Given the plant \eqref{eq:plant} and the transformation matrix $T$,
suppose that the subproblem $F(\ol{N})$ has a feasible solution
for some horizon length $\ol{N}$ and
that the feasible solution makes the objective function have a finite value $V$.
Let $k_0$ be a positive integer larger than or equal to $\ol{N}$
and satisfying $k_0\geq(\lceil V\rceil+1)n_\rma$,
where $\lceil V\rceil$ is the minimum integer larger than or equal to $V$.
On the other hand,
the assumed stabilizability of $(A,B)$ implies
controllability of $(A_\rma,B_\rma)$.
Since the controllability matrix $(B_\rma\ \ A_\rma B_\rma\ \ \cdots\ \
A_\rma^{n_\rma-1}B_\rma)$ has a full row rank,
there exists $g>0$ such that
$\|p^\rmT(B_\rma\ \ A_\rma B_\rma\ \ \cdots\ \ A_\rma^{n_\rma-1}B_\rma)\|_\infty
\geq\|p\|_\infty g$ holds for any vector $p$.
Indeed, the corresponding property is well-known for the $2$-norm,
which is equivalent to the $\infty$-norm
in a finite-dimensional vector space.
With this $g$, let $k_1$ be a nonnegative integer such that
$g>\|A_\rma^{-1}\|_1^{k_1+1}\|B_\rma\|_1$.
Such a $k_1$ exists due to $\|A_\rma^{-1}\|_1<1$.
Now define the zero-input time $K$ by $k_0+k_1$.

\begin{lem}
\label{lem:sparsity}
Given the plant \eqref{eq:plant} and the transformation matrix $T$,
suppose that the subproblem $F(\ol{N})$ is feasible
for some horizon length $\ol{N}$.
Then, if we choose a new horizon length $N$ longer than the zero-input time $K$,
the corresponding finite-horizon subproblem $F(N)$ is feasible and
its optimal solution satisfies
$u^*(k)=0$ for $k=K,\ K+1,\ \ldots,\ N-1$.
Moreover, for any $N'>N$,
the zero extension of this input to the length $N'$ is optimal
in the subproblem $F(N')$.
\end{lem}

\noindent
{\em Proof.}
See Appendix~\ref{pf:sparsity}.
\qed
\bigskip

Looking at Lemma~\ref{lem:sparsity},
it is natural to expect that zero extension of an optimal input
of the finite-horizon subproblem $F(N)$ may give
an optimal input of the infinite-horizon problem $P$.
This is indeed true and will be shown in the next section.

\section{Solution of the Infinite-Horizon Problem}
\label{sec:solution}
Having an optimal input of the finite-horizon subproblem $F(N)$
for some $N>K$ and extending it with the zero input,
we can obtain a feasible input of the infinite-horizon problem $P$.
To see this, note that $x_\rma^*(N)=0$ and $u^*(k)=0$ for $k=N,\ N+1,\ \ldots$,
which implies the state of the anti-stable part, $x_\rma^*(k)$,
remains zero after $k=N$.
On the other hand, the state of the stable part, $x_\rms^*(k)$,
approaches zero after $k=N$ due to the zero input.
For the moment, it is not clear whether the input obtained like this
is optimal or not in the infinite-horizon problem $P$.
We will consider this optimality in this section.

Among feasible inputs for the problem $P$, we need to focus
only on inputs of finite support.

\begin{lem}
\label{lem:finite}
Suppose the infinite-horizon problem $P$ has a feasible solution
that makes the objective function value finite and
write the input of this feasible solution as $u(k)$ $(k=0, 1, \ldots)$.
Then, there exists another feasible input $\ol{u}(k)$ $(k=0,\ 1,\ \ldots)$
that satisfies
$\ol{u}(k)=0$ $(k=K_u,\ K_u+1,\ \ldots)$ for some nonnegative integer $K_u$
and makes the objective function value smaller than or equal to
that of $u(k)$.
\end{lem}

\noindent
{\em Proof.}
See Appendix~\ref{pf:finite}.
\qed
\bigskip

Combining this result with Lemma~\ref{lem:sparsity},
we arrive at the next theorem, which is the main result of this paper.

\begin{thm}
\label{thm:solvability}
Suppose the infinite-horizon problem $P$ has a feasible solution
making an objective function value finite.
Then, the zero-input time $K$ can be defined and,
for a horizon length $N>K$,
the finite-horizon subproblem $F(N)$ has an optimal solution.
Moreover, its zero extension is optimal in the infinite-horizon problem $P$.
\end{thm}

\noindent
{\em Proof.}
Let $u(k)$ $(k=0,\ 1,\ \ldots)$ be any feasible input of the problem $P$
having a finite objective function value.
Lemma~\ref{lem:finite} implies the existence of another feasible input
$\ol{u}(k)$ $(k=0,\ 1,\ \ldots)$ such that
$\ol{u}(k)=0$ for $k=K_u,\ K_u+1,\ \ldots$ and
the objective function value of $\ol{u}(k)$ is
smaller than or equal to that of $u(k)$.
When this input $\ol{u}(k)$ is applied to the plant,
the corresponding state $\ol{x}(k)$ has to go to zero,
which means its anti-stable part has to satisfy $\ol{x}_\rma(K_u)=0$.
This implies that the finite-horizon subproblem $F(K_u)$ is
feasible and thus the assumption of Lemma~\ref{lem:sparsity}
is satisfied for $\ol{N}=K_u$.
Of course, the assumption may be satisfied for a smaller $\ol{N}$.
In any case, we can define the zero-input time $K$
as described before Lemma~\ref{lem:sparsity}.
Lemma~\ref{lem:sparsity} then states that, for any $N>K$,
any optimal input of $F(N)$, denoted by $u^*(k)$ $(k=0, 1, \ldots, N-1)$,
is sparse in the sense that
$u^*(k)=0$ for $k=K,\ K+1,\ \ldots,\ N-1$.
Choose $N$ larger than $K_u$.
If we compare $u^*(k)$ and $\ol{u}(k)$ in the interval $k=0,\ 1,\ \ldots,\ N-1$,
the input $u^*(k)$ makes the objective function value smaller
because of its optimality.
If we extend $u^*(k)$ to the infinite length,
it is feasible in the problem $P$ and makes the objective function value
smaller than or equal to $\ol{u}(k)$ and thus $u(k)$.
Since $u(k)$ is arbitrary, $u^*(k)$ is optimal in $P$.
\qed
\bigskip

This theorem shows how we can obtain an optimal control input for
the infinite-horizon optimal control problem $P$.
That is, first we consider a finite-horizon subproblem $F(N)$ choosing
the horizon length $N$ longer than the zero-input time $K$;
We obtain its optimal solution and consider its zero extension;
Then, this is optimal in the infinite-horizon problem $P$.

The definition of the zero-input time $K$ depends on
the feasibility of the subproblem $F(\ol{N})$.
Noting the convexity of $F(\ol{N})$, we can say the following.
That is, if we choose a polyhedron in the state space and the subproblem
corresponding to each vertex of the polyhedron has a feasible solution
whose objective function value is less than or equal to $V$,
then, for any state in the polyhedron, the corresponding subproblem
has the same property.
This means that the same value of $K$ can be used in this polyhedron
as the zero-input time.
Moreover, suppose we pick up some feasible solution of $P$ and consider
the state trajectory for this feasible control input.
Then for each state in this trajectory
the corresponding subproblem $F(\ol{N})$ is feasible for some $\ol{N}$
and its objective function value does not exceed that of $P$.
This means that we do not need to increase
the zero-input time $K$ as far as the state evolves along this trajectory.

\section{Adaptive Choice of the Horizon Length}
\label{sec:adaptive}
We saw in the previous section that,
if we choose the horizon length $N$ longer than the zero-input time $K$,
an optimal input of a finite-horizon subproblem $F(N)$ gives
an optimal input also in a subproblem with a finite but longer horizon
and then in the infinite-horizon problem $P$ by zero extension.
Although the value of $K$ can be computed as defined in Section~\ref{sec:finite},
this definition may be conservative and may give an unnecessarily large value.
It may be more practical to adaptively choose the horizon length $N$
so that zero extension of an optimal input becomes meaningful.
The next result is useful for this purpose.

\begin{thm}
\label{thm:adaptive}
Suppose that a finite-horizon subproblem $F(N)$ has
an optimal solution whose associated costate satisfies
\[
\|p_\rma^*(N)\|_\infty\|B_\rma\|_1<1.
\]
Then the zero extension of the corresponding optimal input is optimal in $F(N')$
for any $N'>N$ and so is in $P$.
\end{thm}

\noindent
{\em Proof.}
Let $u^*(k)$ $(k=0, 1, \ldots, N-1)$ be the considered optimal input
and let $x_\rma^*(k)$ $(k=0, 1, \ldots, N)$ and
$p_\rma^*(k)$ $(k=1, 2, \ldots, N)$ be the associated state and costate,
respectively.
They satisfy the conditions \eqref{eq:opt_x}--\eqref{eq:opt_u}
by Lemma~\ref{lem:optimality}.
We extend these variables to the length $N'$ by
\begin{align*}
u^*(k)&=0\quad(k=N, N+1, \ldots, N'-1),\\
x_\rma^*(k+1)&=A_\rma x_\rma^*(k)\quad(k=N, N+1, \ldots, N'-1),\\
p_\rma^*(k)^\rmT&=p_\rma^*(k+1)^\rmT A_\rma\quad
(k=N, N+1, \ldots, N'-1).
\end{align*}
They obviously satisfy the optimality conditions
\eqref{eq:opt_x}--\eqref{eq:opt_p}.
If they satisfy in addition the condition \eqref{eq:opt_u}
for $k=N$ to $N'-1$,
this extended input and state form an optimal solution of $F(N')$.

However, for any $k=N, N+1, \ldots, N'-1$, we have
\[
\|p_\rma^*(k+1)^\rmT B_\rma\|_\infty
=\|p_\rma^*(N)^\rmT A_\rma^{-k-1+N} B_\rma\|_\infty
\leq\|p_\rma^*(N)\|_\infty\|A_\rma^{-1}\|_1^{k+1-N}\|B_\rma\|_1
<\|p_\rma^*(N)\|_\infty\|B_\rma\|_1
<1,
\]
which implies $|p_\rma^*(k+1)^\rmT b_i|<1$ for each $i=1, 2, \ldots, m$.
Thus, the condition \eqref{eq:opt_u} is satisfied
with $u^*(k)=0$ for $k=N, N+1, \ldots, N'-1$.

Optimality in $P$ can be shown similarly to Theorem~\ref{thm:solvability}.
\qed
\bigskip

\section{Receding-Horizon Implementation}
\label{sec:receding}
The receding-horizon technique is widely used in practice
in order to utilize a solution of a finite-horizon problem
in the infinite horizon.
It is natural to combine our proposed method with this technique,
which actually exhibits good properties.

We use the receding-horizon technique in the following way.
Suppose that our infinite-horizon problem $P$ has a feasible
solution that makes the objective function value finite.
Formulate its finite-horizon subproblem with the horizon length $N$
longer than the zero-input time $K$.
We write this problem as $F(N,\xi)$ explicitly showing the dependence on
the initial state $\xi$.
Solve this problem and let the obtained optimal input be
$\uind{0}(k)$ $(k=0, 1, \ldots, N-1)$.
Apply its first value $\uind{0}(0)$ to the plant to have the state
$\xind{0}(1)=A\xi+B\uind{0}(0)$.
Then we formulate a new subprogram $F(N,\xind{0}(1))$ with this state
and obtain its optimal solution $\uind{1}(k)$ $(k=1, 2, \ldots, N)$.
Here the time index $k$ is shifted by one in order to be consistent with
the current set-up.
Apply the first value $\uind{1}(1)$ of the input to the plant
to have the next state
$\xind{1}(2)=A\xind{0}(1)+B\uind{1}(1)$.
Then we formulate a new subproblem $F(N,\xind{1}(2))$ for this state
and obtain an optimal control input $\uind{2}(k)$
$(k=2, 3, \ldots, N+1)$.
This procedure is repeated.

The resulting control input is $\uind{k}(k)$ $(k=0, 1, \ldots)$.
It has the following properties.

\begin{thm}
\label{thm:receding}
Suppose our infinite-horizon problem $P$ has a feasible solution
that makes the objective function value finite.
With the horizon length $N$ longer than the zero-input time $K$,
we produce the control input $\uind{k}(k)$ $(k=0, 1, \ldots)$ as described above.
Then, it is optimal in the infinite-horizon problem $P$.
Moreover, it is sparse in the sense that
$\uind{k}(k)=0$ for $k=K, K+1, \ldots$.
\end{thm}

\noindent
{\em Proof.}
Let us begin with $\uind{0}(k)$ $(k=0, 1, \ldots, N-1)$,
which is optimal in $F(N,\xi)$.
We write its objective function value as
\[
V=\|\uind{0}(0)\|_1+\|\uind{0}(1)\|_1+\cdots+\|\uind{0}(N-1)\|_1.
\]
If we extend it by putting zero for one time step, it is optimal
in $F(N+1,\xi)$ thanks to Lemma~\ref{lem:sparsity}.
By the principle of optimality, its tail sequence
$\uind{0}(1), \ldots, \uind{0}(N-1), 0$ is optimal
in $F(N,\xind{0}(1))$.
However, the input $\uind{1}(k)$ $(k=1, 2, \ldots, N)$,
obtained at the time $k=1$,
is also optimal in this same problem, which implies
\[
V=\|\uind{0}(0)\|_1+\|\uind{1}(1)\|_1+\|\uind{1}(2)\|_1+
\cdots+\|\uind{1}(N)\|_1.
\]
This further means that the input
$\uind{0}(0), \uind{1}(1), \uind{1}(2), \ldots, \uind{1}(N)$
is optimal in $F(N+1,\xi)$.

Repeating this discussion, we see that the input
$\uind{0}(0), \uind{1}(1), \ldots,
\uind{K}(K), \uind{K}(K+1), \ldots, \uind{K}(K+N-1)$
is optimal in $F(N+K,\xi)$ and satisfies
\[
V=\|\uind{0}(0)\|_1+\|\uind{1}(1)\|_1+\cdots+
\|\uind{K}(K)\|_1+\|\uind{K}(K+1)\|_1+\cdots+\|\uind{K}(K+N-1)\|_1.
\]
Invoke Lemma~\ref{lem:sparsity} to see that this input is zero
after $k=K$, that is,
$\uind{K}(k)=0$ for $k=K, K+1, \ldots, K+N-1$, which implies
\begin{equation}
V=\|\uind{0}(0)\|_1+\|\uind{1}(1)\|_1+\cdots+\|\uind{K-1}(K-1)\|_1.
\label{eq:recedingvalue}
\end{equation}

It is also possible to show $\uind{k}(k)=0$ for any $k$ larger than $K$
by similar reasoning.
Hence, the receding-horizon input $\uind{k}(k)$ $(k=0, 1, \ldots)$
can be regarded as the zero extension of the finite input sequence
$\uind{k}(k)$ $(k=0, 1, \ldots, K-1)$.
Noting \eqref{eq:recedingvalue} and that $V$ is the achievable optimal value
of $F(N,\xi)$ and also of $P$,
we arrive at the claim of the theorem.
\qed
\bigskip

This theorem says that we can produce an optimal control input
using a receding-horizon technique.
Note that this is not the case in the conventional optimal control
with a quadratic objective function.
The reason is that a value function in a finite horizon changes
its value depending on the horizon length
\cite[Chapter~4]{BGW90}\cite{MaS97}.
In contrast, in our setting,
there is no that change due to the input sparsity
if the horizon length is long enough.
The theorem also guarantees sparsity of the receding-horizon control input.
It is an attractive property because
a naive application of the receding-horizon technique may not
produce a sparse control input as will be seen in Example~\ref{exmpl:simple} below.

The preceding implementation uses the zero-input time $K$
whose definition may be conservative.
We may want to adopt adaptive choice of the horizon length
considered in Section~\ref{sec:adaptive}.
In this case, we choose the horizon length at each step $k$
so that the condition of Theorem~\ref{thm:adaptive} is satisfied.
Let us write this horizon length as $\Nind{k}$.
For each $k$, we consider and solve a finite-horizon subproblem
$F(\Nind{k},\xind{k-1}(k))$ to obtain an optimal control input
$\uind{k}(\ell)$ $(\ell=k, k+1, \ldots, k+\Nind{k}-1)$;
Apply the first value $\uind{k}(k)$ of the input to the plant to have
the next state $\xind{k}(k+1)=A\xind{k-1}(k)+B\uind{k}(k)$;
This is repeated.

\begin{thm}
\label{thm:recedingadaptive}
Suppose our infinite-horizon problem $P$ has a feasible solution
that makes the objective function value finite.
Consider a control input $\uind{k}(k)$ produced by the receding-horizon
technique with the adaptively chosen horizon length $\Nind{k}$.
Then, it is optimal in the infinite-horizon problem $P$.
Moreover, it is sparse in the sense that
$\uind{k}(k)=0$ for $k=K, K+1, \ldots$
with the zero-input time $K$ defined for the initial state $\xi$.
\end{thm}

\noindent
{\em Proof.}
It is proved similarly to the previous theorem
with Theorem~\ref{thm:adaptive}.
\qed
\bigskip

Here we go back to our problem formulation in Section~\ref{sec:problem}
and consider why $\alpha>1$ is necessary when $A_\rmo$,
the original system matrix in \eqref{eq:original},
has an eigenvalue on the unit circle.
Suppose a simple plant $A_\rmo=1$ and $B_\rmo=1$ with the initial state $\xi=1$
and consider the problem $P_\rmo$.
If we reformulate it to the problem $P$ with $\alpha=1$,
we have $x(k)=\xi+u(0)+u(1)+\cdots+u(k-1)$.
Hence, if an input $u(k)$ is equal to $-1$ for one $k$ and equal to zero
otherwise, it is optimal.
This means that the time to apply a nonzero input can be postponed
indefinitely and thus
there is no zero-input time such as in Lemma~\ref{lem:sparsity}.
Moreover, it is unlikely that a receding-horizon technique is useful
in this situation because the technique heavily relies on the first value
in the input sequence.
We need to prohibit such postponement to recover the results in this paper
and one way is to put penalty on the future input by choosing $\alpha>1$.

\section{Examples}
\label{sec:examples}
We applied the proposed approach to two example problems.

\begin{exmpl}
\label{exmpl:simple}
Consider as a plant a simple continuous-time system:
\[
\dot{x}_\rmc(t)=\begin{pmatrix} 0 & 1 \\ 1 & 0 \end{pmatrix}x_\rmc(t)
+\begin{pmatrix} 0 \\ 1 \end{pmatrix}u_\rmc(t),\quad
x_\rmc(0)=\begin{pmatrix} 2.3 \\ -2.6 \end{pmatrix}.
\]
For its control, it was discretized
with the zero-order hold for the sampling period $h=0.01$.
With the correspondence $x_\rmo(k)=x_\rmc(kh)$,
the resulting discretized plant was expressed in the form of
\eqref{eq:original} with
\[
A_\rmo=\exp\Bigl[
\begin{pmatrix} 0 & 1 \\ 1 & 0 \end{pmatrix}\! h\Bigr],\quad
B_\rmo=\int_0^h
\exp\Bigl[
\begin{pmatrix} 0 & 1 \\ 1 & 0 \end{pmatrix}\! t\Bigr]\rmd t
\begin{pmatrix} 0 \\ 1 \end{pmatrix},\quad
\xi=\begin{pmatrix} 2.3 \\ -2.6 \end{pmatrix}.
\]

\begin{figure}[t]
\begin{center}
\setlength{\unitlength}{1cm}
\begin{picture}(0,5)
\put(0,0){\makebox(0,0)[b]{\includegraphics[height=5cm]{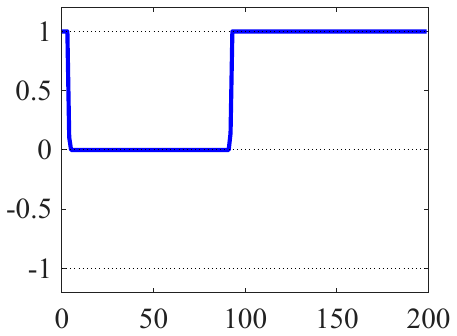}}}
\put(0.3,-0.3){\makebox(0,0){time $k$}}
\end{picture}
\end{center}
\caption{An optimal input of a finite-horizon optimal control problem.
It is equal to zero for $k=5, 6, \ldots, 91$.
The objective function value
(\ie, the sum of the input norms) is $111.2$.}
\label{fig:finite}
\end{figure}

\begin{figure}[t]
\begin{center}
\setlength{\unitlength}{1cm}
\begin{picture}(0,5)
\put(0,0){\makebox(0,0)[b]{\includegraphics[height=5cm]{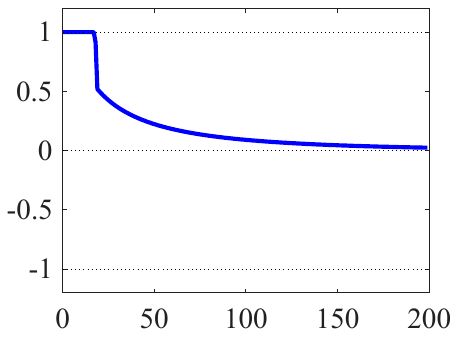}}}
\put(0.3,-0.3){\makebox(0,0){time $k$}}
\end{picture}
\end{center}
\caption{An input generated by naive application of the
receding-horizon technique.
It is not equal to zero at least up to $k=199$, thus, not sparse.
The sum of the input norms is $41.8$ up to $k=199$.}
\label{fig:naive}
\end{figure}

We first took a naive approach for its sparse control.
Set the horizon length as $N=200$ and consider the optimization problem:
\begin{alignat*}{2}
&\text{minimize}\ \ &&\sum_{k=0}^{N-1}\|u_\rmo(k)\|_1\\
&\text{subject to}\ \
&&x_\rmo(k+1)=A_\rmo x_\rmo(k)+B_\rmo u_\rmo(k)\ \ (k=0,\ 1,\ \ldots,\ N-1),\\
&&&x_\rmo(0)=\xi,\quad x_\rmo(N)=0,\\
&&&\|u_\rmo(k)\|_\infty\leq 1\ \ (k=0,\ 1,\ \ldots,\ N-1).
\end{alignat*}
Note that this problem is concerned with both stable part and anti-stable
part unlike our subproblem $F(N)$.
Figure~\ref{fig:finite} shows an optimal input for this problem,
which is equal to zero for $k=5, 6, \ldots, 91$.

In order to use this input in the infinite horizon,
we adopted the receding-horizon technique fixing the horizon length to $200$.
The produced input is presented in Figure~\ref{fig:naive}.
It is not sparse at all.
Note that, even if an input sequence produced at each $k$ is sparse,
its first value, actually used in the receding-horizon technique,
may not be equal to zero.

We next applied the proposed approach.
Since our matrix $A_\rmo$ had no eigenvalues on the unit circle,
we set $\alpha=1$, which meant $x(k)=x_\rmo(k)$, $u(k)=u_\rmo(k)$,
$A=A_\rmo$, and $B=B_\rmo$.
The plant was decomposed into the stable part and
the anti-stable part with the transformation matrix
$T=\begin{pmatrix} 1 & 1 \\ -1 & 1 \end{pmatrix}$.
The subproblem $F(\ol{N})$ became feasible for $\ol{N}=36$ and
gave the objective function value $V=35.7$.
The zero-input time $K$ in Lemma~\ref{lem:sparsity}
was computed as $37$.

\begin{figure}[t]
\begin{center}
\setlength{\unitlength}{1cm}
\begin{picture}(0,5)
\put(0,0){\makebox(0,0)[b]{\includegraphics[height=5cm]{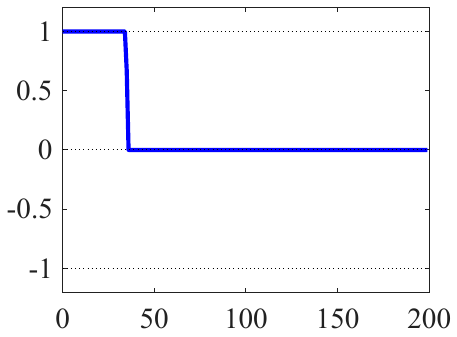}}}
\put(0.3,-0.3){\makebox(0,0){time $k$}}
\end{picture}
\end{center}
\caption{An input given by the proposed approach.
It is equal to zero after $k=36$ and makes the
objective function value $35.7$,
which is smaller than the one given by the receding-horizon
input in Figure~\ref{fig:naive}.}
\label{fig:antistable}
\end{figure}

We solved the finite-horizon subproblem $F(N)$ for $N=38$
and considered zero extension of an optimal input.
The resulting input is presented in Figure~\ref{fig:antistable}.
By Theorem~\ref{thm:solvability}, this input is optimal
in the infinite horizon.
Indeed, it makes the objective function value $35.7$,
which is smaller than $41.8$,
a lower bound of the objective function value achieved by
the previous receding-horizon input (the sum of the input norms up to $k=199$).
We can also see that the input is sparse.
Theorem~\ref{thm:solvability} guarantees that the input is equal to zero
after $k=K=37$.
In fact, it becomes zero earlier at $k=36$.

Our approach was used also with the receding-horizon technique.
The generated input looked
close to the one in Figure~\ref{fig:antistable}.
\fin
\end{exmpl}

\begin{exmpl}
\label{exmpl:HCW}
We considered a motion of a chaser spacecraft relative to
a target, which was in a circular orbit around the earth.
This motion is approximately described
by the Hill--Clohessy--Wiltshire equation
\cite[Section~4.6]{Wie08}\cite[Section~6.8]{Val22}:
\begin{equation}
\frac{\rmd}{\rmd t}
\begin{pmatrix} a(t) \\ b(t) \\ \dot{a}(t) \\ \dot{b}(t) \end{pmatrix}
=
\begin{pmatrix} 0 & 0 & 1 & 0 \\ 0 & 0 & 0 & 1 \\
3\omega^2 & 0 & 0 & 2\omega \\ 0 & 0 & -2\omega & 0 \end{pmatrix}
\begin{pmatrix} a(t) \\ b(t) \\ \dot{a}(t) \\ \dot{b}(t) \end{pmatrix}
+
\begin{pmatrix} 0 & 0 \\ 0 & 0 \\ 1 & 0 \\ 0 & 1 \end{pmatrix}
\begin{pmatrix} u_a(t) \\ u_b(t) \end{pmatrix}.
\label{eq:HCW}
\end{equation}
Here, $(a(t), b(t))$ stands for the relative position of the chaser
with respect to the target in the plane containing the orbit of the
target,
where the $a$-axis is in the radial direction from the earth to the target
and the $b$-axis in the tangential direction in which the target travels.
The motion perpendicular to the orbital plane is not considered here.
The inputs $u_a(t)$ and $u_b(t)$ are the force per unit mass of the chaser
along the $a$- and $b$-axes, respectively.
The quantity $\omega$ is equal to $\sqrt{GM/r^3}$
for the gravitational constant $G$, the mass of the earth $M$, and
the radius $r$ of the orbit of the target.
We chose $r=6.8\times 10^6\,[\text{m}]$,
which gave $\omega=1.1\times 10^{-3}\,[\rms^{-1}]$.

In order to formulate the problem into the form of $P_\rmo$,
we discretized the system with the sampling period $h=1\,[\rms]$
and wrote it as $x_\rmo(k+1)=A_\rmo x_\rmo(k)+B_\rmo u_\rmo(k)$.
The initial state was chosen as $\xi=[100\ \ 100\ \ 0\ \ 0]^\rmT$
in meters.
The matrix $A_\rmo$ had all of its eigenvalues on the unit circle.
We hence chose $\alpha=1.05$ and set up the problem $P$
with $A=\alpha A_\rmo$ and $B=\alpha B_\rmo$.

\begin{figure}[t]
\begin{center}
\setlength{\unitlength}{1cm}
\begin{picture}(0,5)
\put(0,0){\makebox(0,0)[b]{\includegraphics[height=5cm]{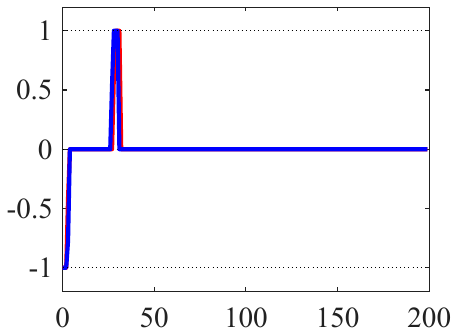}}}
\put(0.3,-0.3){\makebox(0,0){time $k$}}
\end{picture}
\end{center}
\caption{An optimal input $u^*_\rmo(k)$ for the spacecraft example,
which was computed with an adaptively chosen horizon length $N=172$.
The first component of the input is expressed by the red line and
the second by the blue line.
They both appear to equal to zero after $k=32$.}
\label{fig:HCW}
\end{figure}

We evaluated the zero-input time $K$.
The transformation matrix $T$ was chosen as
\[
T=\begin{pmatrix}
0 & -2/3\omega & 0 & -1 \\ 1 & 0 & 2 & 0 \\
0 & 0 & \omega & 0 \\ 0 & 1 & 0 & 2\omega \end{pmatrix}
\begin{pmatrix} 1 & 0 & 0 & 0 \\ 0 & 0.01 & 0 & 0 \\
0 & 0 & 1 & 0 \\ 0 & 0 & 0 & 1 \end{pmatrix},
\]
which gave $\|T^{-1}A^{-1}T\|_1=0.96$.
Here the left matrix in the definition of $T$
is a matrix that converts the original system matrix in \eqref{eq:HCW}
into the Jordan canonical form and
the right matrix is for reduction of the off-diagonal element.
For the horizon length $\ol{N}=40$,
the subproblem $F(\ol{N})$ had a feasible solution with the objective
function value $V=37.5$.
With this preparation, the zero-input time $K$ was computed as $400$.
In order to shorten the required horizon length,
we tried the adaptive choice described in Section~\ref{sec:adaptive}
and found the horizon length $N=172$ satisfied the required condition.
Figure~\ref{fig:HCW} shows the resulting optimal input $u^*_\rmo(k)$,
where the effect of the factor $\alpha$ has been removed.
Both of its two components equal to zero after $k=32$.
We also applied the receding-horizon technique.
Although we checked the condition of Theorem~\ref{thm:recedingadaptive}
at every time step $k$,
there was no need to change the value of $N=172$.
The resulting input was almost the same as the one in Figure~\ref{fig:HCW}.
\fin
\end{exmpl}

\section{Conclusion}
\label{sec:concl}
We considered in this paper sparse optimal control in the infinite horizon.
The objective function here is the sum of the input $1$-norms,
which makes the optimal input sparse.
Thanks to this property, an optimal input in the infinite horizon
can be obtained by zero extension of an optimal input
of a finite-horizon subproblem $F(N)$.
Here, the horizon length $N$ is chosen longer than
the zero-input time $K$.
Explicit computation of $K$ is possible and adaptive choice of
the horizon length $N$ is also discussed.
The proposed approach can be used with the receding-horizon technique
and gives an optimal and sparse input in the infinite horizon.

The objective function employed here is the sum of the input $1$-norms.
Although this objective function is widely used in literature,
it is natural to ask if a similar approach is possible to a more
general class of objective functions.
The research is proceeding in this direction and the result
should be presented in the near future.

\appendix
\section{Proof of the Sufficiency Part of Lemma~\ref{lem:optimality}}
\label{pf:optimality}
Let $u^*(k)$, $x_\rma^*(k)$, and $p_\rma^*(k)$ satisfy the conditions
\eqref{eq:opt_x}--\eqref{eq:opt_u}.
On the other hand, let $u(k)$ and $x_\rma(k)$ be any input and state
that satisfy
\begin{align*}
x_\rma(k+1)&=A_\rma x_\rma(k)+B_\rma u(k)\ \
(k=0, 1, \ldots, N-1),\\
x_\rma(0)&=(O\ \ I)T^{-1}\xi,\ \ x_\rma(N)=0,\\
\|u(k)\|_\infty&\leq\alpha^k\ \
(k=0, 1, \ldots, N-1).
\end{align*}
Then we have
\begin{align*}
\sum_{k=0}^{N-1}\|u(k)\|_1
&=\sum_{k=0}^{N-1}\|u(k)\|_1+
\sum_{k=0}^{N-1}p_\rma^*(k+1)^\rmT\Bigl[A_\rma x_\rma(k)+B_\rma u(k)-x_\rma(k+1)\Bigr]\\
&=\sum_{k=0}^{N-1}\sum_{i=1}^m
\Bigl[|u_i(k)|+p_\rma^*(k+1)^\rmT b_i u_i(k)\Bigr]
+\sum_{k=1}^{N-1}\Bigl[p_\rma^*(k+1)^\rmT A_\rma-p_\rma^*(k)^\rmT\Bigr]x_\rma(k)\\
&\quad+p_\rma^*(1)^\rmT A_\rma x_\rma(0)-p_\rma^*(N)^\rmT x_\rma(N).
\end{align*}
The first equality follows from $x_\rma(k+1)=A_\rma x_\rma(k)+B_\rma u(k)$ and
the second from the change of the summation order.
In the last expression,
the first term is minimized at $u^*(k)$ due to \eqref{eq:opt_u}
and the remaining terms are constant irrespective of $x_\rma(k)$
due to \eqref{eq:opt_p}.
Therefore, it is minimized at $u^*(k)$.

\section{Proof of Lemma~\ref{lem:sparsity}}
\label{pf:sparsity}
%
%
%
For any $N>K$, we consider the subproblem $F(N)$ and
its optimal input $u^*(k)$ $(k=0,\ 1,\ \ldots,\ N-1)$.
This $N$ is larger than $\ol{N}$,
which was used in the definition of $K$,
and the feasible solution for $F(\ol{N})$ gives a feasible
solution of the present problem $F(N)$ by zero extension.
Hence, the present optimal solution $u^*(k)$ has
the objective function value smaller than or equal to $V$.

Recall $k_0$ that appeared in the definition of $K$.
We show the existence of an integer $k_2\leq k_0$ such that
$|u_i^*(k)|<\alpha^k$ for any $i=1,\ 2,\ \ldots, m$ and any $k=k_2-1,\
k_2-2,\ \ldots,\ k_2-n_\rma$.
Indeed, among $\lceil V\rceil+1$ intervals
$k_0-1\geq k\geq k_0-n_\rma$, $k_0-n_\rma-1\geq k\geq k_0-2n_\rma$, $\ldots$,
$k_0-\lceil V\rceil n_\rma-1\geq k\geq k_0-(\lceil V\rceil+1)n_\rma$,
at least one interval should have the property above.
If this is not the case, each interval has $k$
such that $|u_i^*(k)|\geq\alpha^k$ for some $i$,
which means $\|u^*(k)\|_1\geq\alpha^k\geq 1$.
Then, the value of the objective function satisfies
$\sum_{k=0}^{N-1}\|u^*(k)\|_1\geq\lceil V\rceil+1>V$,
which is a contradiction.

For the $k_2$ above, Lemma~\ref{lem:optimality} implies
\[
|p_\rma^*(k+1)^\rmT b_i|\leq 1
\quad(i=1,\ 2,\ \ldots, m;\ \ k=k_2-1,\ k_2-2,\ \ldots,\ k_2-n_\rma).
\]
Noting
$p_\rma^*(k_2)^\rmT A_\rma=p_\rma^*(k_2-1)^\rmT$,
$p_\rma^*(k_2-1)^\rmT A_\rma=p_\rma^*(k_2-2)^\rmT$, $\ldots$, we have
\[
|p_\rma^*(k_2)^\rmT b_i|\leq 1,\ \
|p_\rma^*(k_2)^\rmT A_\rma b_i|\leq 1,\ \
\ldots,\ \
|p_\rma^*(k_2)^\rmT A_\rma^{n_\rma-1}b_i|\leq 1
\]
for any $i=1,\ 2,\ \ldots,\ m$.
This leads to
\[
1\geq
\bigl\|p_\rma^*(k_2)^\rmT(B_\rma\ \ A_\rma B_\rma\ \
\cdots\ \ A_\rma^{n_\rma-1}B_\rma)\bigr\|_\infty.
\]
By the definition of the positive number $g$, we have
$1\geq\|p_\rma^*(k_2)\|_\infty g$.

Now for any integer $k$ such that $K=k_0+k_1\leq k\leq N-1$
and any $i=1,\ 2,\ \ldots,\ m$,
\[
|p_\rma^*(k+1)^\rmT b_i|
=|p_\rma^*(k_2)^\rmT A_\rma^{-(k-k_2+1)}b_i|
\leq\|p_\rma^*(k_2)\|_\infty \|A_\rma^{-1}\|_1^{k-k_2+1}\|b_i\|_1
\leq\frac{1}{g}\|A_\rma^{-1}\|_1^{k-k_2+1}\|b_i\|_1.
\]
Since $k-k_2+1\geq k-k_0+1\geq k_1+1$,
the above value is less than $1$ and
thus Lemma~\ref{lem:optimality} implies $u_i^*(k)=0$
for any $i=1,\ 2,\ \ldots,\ m$.

In order to show the second statement, extend the state and the costate
to the length $N'$ so that
\begin{align*}
x_\rma^*(k+1)&=A_\rma x_\rma^*(k)\ \
(k=N,\ N+1,\ \ldots,\ N'-1),\\
p_\rma^*(k)^\rmT&=p_\rma^*(k+1)^\rmT A_\rma\ \
(k=N,\ N+1,\ \ldots,\ N'-1).
\end{align*}
These extended state and costate satisfy
the optimality conditions \eqref{eq:opt_x}--\eqref{eq:opt_p}
in Lemma~\ref{lem:optimality}
together with the input by zero extension.
They also satisfy the condition \eqref{eq:opt_u} because
$|p_\rma^*(k+1)^\rmT b_i|<1$
for $k=N, N+1, \ldots, N'-1$, which is similarly seen to the above.
Hence Lemma~\ref{lem:optimality} assures the optimality of
this extended input.

\section{Proof of Lemma~\ref{lem:finite}}
\label{pf:finite}
The assumed stabilizability of $(A,B)$ implies controllability of
the anti-stable part $(A_\rma,B_\rma)$.
The corresponding controllability matrix
$W_\rma=(B_\rma\ \ A_\rma B_\rma\ \ \cdots\ \ A_\rma^{n_\rma-1}B_\rma)$
hence has a full row rank.
For the assumed feasible input $u(k)$, consider its corresponding
state $x(k)$ together with its anti-stable part $x_\rma(k)$.
Since the input $u(k)$ has a finite objective function value,
$u(k)\rightarrow 0$ as $k\rightarrow\infty$.
This implies the existence of $\ol{k}$ such that
$\|u(k)\|_\infty\leq 1/2$
for any $k=\ol{k},\ \ol{k}+1,\ \ldots,\ \ol{k}+n_\rma-1$.
Now choose a positive integer $K_u>\ol{k}+n_\rma-1$ so that
each element of the vector defined by
\[
(\delta_{n_\rma-1}^\rmT\ \ \delta_{n_\rma-2}^\rmT\ \ \cdots\ \
\delta_0^\rmT)^\rmT=
-W_\rma^\rmT(W_\rma W_\rma^\rmT)^{-1}A_\rma^{-(K_u-\ol{k}-n_\rma)}x_\rma(K_u)
\]
has a magnitude less than or equal to $1/2$ and
\[
\|W_\rma^\rmT(W_\rma W_\rma^\rmT)^{-1}\|_1\|A_\rma^{-1}\|_1^{K_u-\ol{k}-n_\rma}
\leq\frac{1}{\|B_\rma\|_1}.
\]
Such a $K_u$ exists because the matrix $A_\rma^{-1}$ has its eigenvalues
inside the unit circle and satisfies $\|A_\rma^{-1}\|_1<1$.
Note also that $x(k)$ and $x_\rma(k)$ converge to zero
due to the feasibility of $u(k)$.

Define a new input $\ol{u}(k)$ by
\[
\ol{u}(k)=
\begin{cases}
u(k)+\delta_j
&\text{ for $k=\ol{k}+j$\ \ $(j=0,\ 1,\ \ldots,\ n_\rma-1)$},\\
0&\text{ for $k=K_u,\ K_u+1,\ \ldots$,}\\
u(k)&\text{ otherwise.}
\end{cases}
\]
Recall that each element of $\delta_j$ has the magnitude less than
or equal to $1/2$, which means $\|\ol{u}(k)\|_\infty\leq\alpha^k$ for any $k$.
When the input $\ol{u}(k)$ is applied to the plant,
the corresponding state is written as $\ol{x}(k)$ and
its stable part and anti-stable part as $\ol{x}_\rms(k)$ and
$\ol{x}_\rma(k)$, respectively.
Here we have
\begin{align*}
\ol{x}_\rma(K_u)
&=A_\rma^{K_u}x_\rma(0)+\sum_{k=0}^{K_u-1}A_\rma^{K_u-k-1}B_\rma\ol{u}(k)\\
&=A_\rma^{K_u}x_\rma(0)+\sum_{k=0}^{K_u-1}A_\rma^{K_u-k-1}B_\rma u(k)
+\sum_{j=0}^{n_\rma-1}A_\rma^{K_u-\ol{k}-j-1}B_\rma\delta_j\\
&=x_\rma(K_u)+A_\rma^{K_u-\ol{k}-n_\rma}
(B_\rma\ \ A_\rma B_\rma\ \ \cdots\ \ A_\rma^{n_\rma-1}B_\rma)
(\delta_{n_\rma-1}^\rmT\ \ \delta_{n_\rma-2}^\rmT\ \ \cdots\ \
\delta_0^\rmT)^\rmT\\
&=0.
\end{align*}
After $k=K_u$ the anti-stable part of the state, $\ol{x}_\rma(k)$,
constantly equals to zero
because so does the input $\ol{u}(k)$.
The stable part of the state, $\ol{x}_\rms(k)$, converges to zero
after $k=K_u$ due to the zero input.
Hence the whole state $\ol{x}(k)$ converges to zero, which means that
the input $\ol{u}(k)$ is feasible in the problem $P$.

It remains to show that $\ol{u}(k)$ does not make
the objective function value larger than $u(k)$.
For that, it suffices to show
\[
\|(\delta_{n_\rma-1}^\rmT\ \ \delta_{n_\rma-2}^\rmT\ \cdots\ \
\delta_0^\rmT)^\rmT\|_1
\leq\sum_{k=K_u}^\infty\|u(k)\|_1.
\]
Note first that
\begin{align}
\|(\delta_{n_\rma-1}^\rmT\ \ \delta_{n_\rma-2}^\rmT\ \ \cdots\ \
\delta_0^\rmT)^\rmT\|_1
&\leq\|W_\rma^\rmT(W_\rma W_\rma^\rmT)^{-1}\|_1
\|A_\rma^{-1}\|_1^{K_u-\ol{k}-n_\rma}\|x_\rma(K_u)\|_1
\nonumber\\
&\leq\frac{1}{\|B_\rma\|_1}\|x_\rma(K_u)\|_1.
\label{eq:1}
\end{align}
On the other hand, for any $\ell>K_u$,
\[
x_\rma(\ell)
=A_\rma^{\ell-K_u}x_\rma(K_u)+\sum_{k=K_u}^{\ell-1}A_\rma^{\ell-k-1}B_\rma u(k),
\]
which implies
\[
A_\rma^{K_u-\ell}x_\rma(\ell)=
x_\rma(K_u)+\sum_{k=K_u}^{\ell-1}A_\rma^{K_u-k-1}B_\rma u(k).
\]
In the limit of $\ell\rightarrow\infty$,
$A_\rma^{K_u-\ell}$ converges to zero and so does
$x_\rma(\ell)$,
which gives
\[
0=x_\rma(K_u)+\sum_{k=K_u}^\infty A_\rma^{K_u-k-1}B_\rma u(k).
\]
Thus we have
\begin{equation}
\|x_\rma(K_u)\|_1
\leq\sum_{k=K_u}^\infty\|A_\rma^{-1}\|_1^{k+1-K_u}\|B_\rma\|_1\|u(k)\|_1
\leq\|B_\rma\|_1\sum_{k=K_u}^\infty\|u(k)\|_1.
\label{eq:2}
\end{equation}
The inequalities~\eqref{eq:1} and \eqref{eq:2} give the desired inequality.

\end{document}